\documentclass{article}

\usepackage{PRIME1arxiv}

\usepackage[
backend=biber,
style=alphabetic,
sorting=ynt
]{biblatex}

\usepackage{biblatex}

\usepackage[utf8]{inputenc} 
\usepackage[T1]{fontenc}    
\usepackage{hyperref}       
\usepackage{url}            
\usepackage{booktabs}       
\usepackage{amsfonts}       
\usepackage{nicefrac}       
\usepackage[utf8]{inputenc} 
\usepackage[T1]{fontenc}    
\usepackage{hyperref}       
\usepackage{url}            
\usepackage{booktabs}       
\usepackage{amsfonts}       
\usepackage{amsmath}
\usepackage{bbm}
\usepackage{amsthm}
\usepackage{mathtools}
\usepackage{nicefrac}       
\usepackage{microtype}      
\usepackage{fancyhdr}       
\usepackage{graphicx}       
\usepackage{wrapfig}
\graphicspath{{media/}}     
\usepackage{dsfont}
\usepackage{xcolor}
\usepackage{algorithm}
\usepackage{algpseudocode}
\usepackage[most]{tcolorbox}
\usepackage{mathtools}
\graphicspath{{media/}}     
\usepackage{tikz}
\usepackage{tabularx}
\usetikzlibrary{positioning, chains,fit, shapes,calc}

\newtheorem{theorem}{Theorem}
\makeatletter
\def\algbackskip{\hskip-\ALG@thistlm}
\makeatother
\algnewcommand\algorithmicforeach{\textbf{for each}}
\algdef{S}[FOR]{ForEach}[1]{\algorithmicforeach\ #1\ \algorithmicdo}

\newcommand{\supp}[1]{\textsc{supp}(#1)}
\newcommand{\test}[2]{\textsc{Test}(#1, #2)}
\newcommand{\R}{\mathbb{R}}
\newcommand{\Z}{\mathbb{Z}}
\newcommand{\F}{\mathbb{F}}
\newcommand{\wt}{\textsc{wt}}
\newcommand{\ti}[1]{\tilde{#1}}

\newcommand{\OR}{\textsc{OR }}
\newcommand{\f}[1]{\mathbf{#1}}
\newcommand{\corr}{\textsc{corr}}

\newcommand{\GreedyComp}{\textsc{Greedy Complete}}
\newcommand{\CheckCov}{\textsc{Check Coverage}}
\newcommand{\DisjDec}{\textsc{Disjunct Matrix Decoding}}

\newtheorem{corollary}[theorem]{Corollary}
\newtheorem{lemma}[theorem]{Lemma}

\newtheorem{definition}[theorem]{Definition}

\newcommand{\gv}[1]{{\color{red} #1}}

\title{Combinatorial Group Testing in Presence of Deletions
}

\author{
 Venkata Gandikota \\
 EECS, Syracuse University \\
 Syracuse, NY \\
 \texttt{vsgandik@syr.edu} \\
  \And
  Nikita Polyanskii\\
  IOTA Foundation\\
  Berlin, Germany\\
  \texttt{nikita.polianskii@iota.org}
  \And
 Haodong Yang \\
 EECS, Syracuse University \\
 Syracuse, NY\\
 \texttt{hyang85@syr.edu}
}

\begin{document}
\definecolor{myblue}{RGB}{80,80,160}
\definecolor{mygreen}{RGB}{80,160,80}
\maketitle
\begin{abstract}
 The study of group testing aims to develop strategies to identify a small set of defective items among a large population using a few pooled tests. The established techniques have been highly beneficial in a broad spectrum of applications ranging from channel communication to identifying COVID-19-infected individuals efficiently. Despite significant research on group testing and its variants since the 1940s, testing strategies robust to deletion noise have not been explored. Deletion errors, common in practical systems like wireless communication and data storage, cause asynchrony in tests, rendering current group testing methods ineffective.
In this work, we introduce non-adaptive group testing strategies resilient to deletion noise. We establish the necessary and sufficient conditions for successfully identifying defective items despite adversarial deletions of test outcomes. The study also presents constructions of testing matrices with a near-optimal number of tests and develops efficient and super-efficient recovery algorithms.
\end{abstract}


\section{Introduction}
The study of combinatorial group testing was initated by Dorfman \cite{dorfman1943detection} to devise screening strategies for the identification of syphilis-infected soldiers during World War II. The main observation that fueled this area of research is that pooling blood samples from different soldiers together can provide a significant reduction in the number of blood tests required to identify a few infected individuals. A pooled test containing blood samples from different individuals will be positive if and only if it contains at least one infected sample. Hence, compared to individual testing where a negative test just confirms a single healthy person, a single pooled test would mark all the sampled individuals within the test as healthy. Since then, group testing techniques have found applications in a wide range of domains such as industrial quality assurance \cite{sobel1959group}, DNA sequencing \cite{pevzner1994towards}, molecular biology \cite{macula1999probabilistic, schliep2003group}, wireless communication \cite{wolf1985born,berger1984random,luo2008neighbor}, data compression \cite{hong2002group,emad2014poisson}, pattern matching \cite{clifford2007k}, secure key distribution \cite{stinson2000secure}, network tomography \cite{ma2014node,cheraghchi2012graph}, efficient data analysis \cite{engels2021practical, cormode2005s} and more recently in COVID-19 testing protocols \cite{wang2023tropical,ghosh2021compressed}. 

Formally, the central question in group testing is to identify a small unknown set $S^*$ of $k $ defective items among a population of $n (\gg k)$ using a few grouped tests. The outcome of a pooled test corresponding to a subset $P \subseteq [n]$ of items is positive if it contains at least one defective element, i.e., $S^* \cap P \neq \emptyset$. Theoretical research in group testing has focused on the following central problem: 
\emph{How many pooled tests $\{P_1, \ldots, P_m\} \subseteq [n]$ are required to recover $S^*$ from the test outcomes of the form $\mathbbm{1}_{S^* \cap P \neq \emptyset}$?} The algorithmic question is then to efficiently decode to recover the identities of all the infected individuals, $S^*$ given all the test outcomes. 

Decades of research have helped us answer these questions in group testing and many of its variants such as \emph{probabilistic group testing} \cite{guruswami2023noise, cheraghchi2011group,aldridge2019group}, where the possible set of defectives have an associated prior, \emph{quantitative testing} \cite{hao1990optimal,gargano1992improved,d2013superimposed,inan2019optimality} where the test outcome is non-binary such as $|S^*\cap P|$, and \emph{constrained group testing} \cite{cheraghchi2012graph,gandikota2019nearly,inan2019sparse,agarwal2020group} where the allowable tests have some constraints on them. See \cite{d2014lectures,malyutov2013search,aldridge2019group,ngo2000survey,chen2008survey} for excellent surveys on established group testing techniques. 

Our focus in this work is on \emph{combinatorial group testing} (CGT), where the set of defectives can be any arbitrary $k$-subset of $[n]$. 
Furthermore, we will restrict our attention to non-adaptive tests, where the pooling strategy is fixed in advance. This is in contrast to adaptive testing, where each subsequent test is designed after seeing the previous test outcomes. Non-adaptive tests are more practical as they allow independent parallel testing thereby significantly speeding up the detection process.  It is, by now, folklore that $m=\Omega(k^2 \log n / \log k)$ non-adaptive tests are required for the exact recovery of the defectives in this model \cite{furedi1996onr, ruszinko1994upper}. Pooling strategies based on disjunct matrices enable efficient recovery using $m=O(k^2 \log n)$ non-adaptive tests in $O(mn)$ time \cite{du2000combinatorial,d1982bounds}. Several recent works have also devised super-efficient sublinear time algorithms that run in $\tilde{O}(m)$ time \cite{indyk2010efficiently,cheraghchi2020combinatorial,guruswami2023noise, cai2013grotesque}.

In practice, we observe that the test outcomes are often noisy. This noise can manifest in various forms, however, the most commonly encountered noise model in group testing is the bit-fip noise, where the binary test outcome can be flipped from a positive to a negative or vice-versa. 
The bit-flip noise model in group testing has been well-understood in both random \cite{guruswami2023noise, atia2012boolean, chan2011non,atia2009noisy,mazumdar2016nonadaptive,cai2013grotesque} and adversarial \cite{cheraghchi2009noise,ngo2011efficiently,li2021combinatorial} settings. In particular, the results of \cite{cheraghchi2009noise} establish bounds on the maximum fraction of tolerable false positives and false negatives in the test outcomes for exact and approximate recovery, and also devise pooling strategies robust to such errors. Despite decades of intense research in group testing, other noise models such as deletion noise are not well-understood. 

The problem is quite surprisingly unexplored as deletion errors are known to be prevalent in many applications such as wireless communications \cite{chen2020correcting}, and data storage \cite{chauhan2021portable,cheraghchi2020coded} where group testing techniques are widely used. 
Consider a variant of the A-channel \cite{Lancho2022Allerton} and the Binary-Adder Channel (BAC) \cite{fan1995superimposed}, where a set of \( K \) out of \( N \) users are active at any given time, and each active user transmits a codeword \( c_j \in C \) from a binary codebook \( C \subseteq \mathbb{F}_2^n \). The channel output \( Y \) can be modeled as \( Y = \bigvee_{j \in [K]} c_j \). For comparison, in an A-channel, the output is modeled as \( Y = \bigcup_j c_j \), whereas the BAC models the output as \( Y = \sum_j c_j \).  
The goal in such settings is to recover the set of transmitting users at a specific instant given the channel output \( Y \). However, factors like channel congestion and packet drops can introduce asynchronous outcomes at the receiver, resulting in errors that can be modeled as deletions.  
The decoding process for this channel can be framed as a group testing problem, where the codewords represent the identities of individual users.  
The constructions proposed in this paper can be effectively utilized to design codes for such channels, enabling efficient decoding, and improving our understanding of deletion resilient superimposed codes. 

In the deletion noise model, certain test outputs are deleted which shifts the outcomes, thereby causing asynchrony between the devised tests and their results. Note that deletion errors are in general more difficult to handle than bit-flip errors since even one deletion out of $m$ tests can be equivalent to $O(m)$ bit flips in the first $m-1$ tests. For instance, if the true output sequence of $m$ tests is given by $(1,0,1,0, \ldots, 1)$, where $1$ denotes a positive outcome, and $0$ denotes a negative output, then deleting the first test outcome will result in an $m-1$ length test output sequence $(0,1,0,1,\ldots,1)$ which has a large Hamming distance from the true outputs. The main difficulty in recovering $S^*$ from deletion-inflicted test outputs stems from this mismatch between the tests and their true outcomes. Through this work, we aim to enhance our understanding of noise-resilient group testing with adversarial deletions. 

We represent the set of $m$ tests as a binary testing matrix $A \in \{0,1\}^{m \times n}$, where each row indicates the set of elements pooled together in a particular test. We start by characterizing the necessary and sufficient conditions on the testing matrix for recovery in the presence of deletion noise. The necessary conditions are established by generalizing the notion of separable matrices -- the necessary conditions for exact recovery in noiseless group testing. Specifically, to be able to distinguish between any two distinct sets of defectives, we need their test outcomes to be different even after some $\Delta$ deletions. This leads to the definition of a \emph{deletion-separable matrix} that also provides a simple lower bound on the number of required tests. However, the best-known decoding algorithm given test outcomes with respect to a deletion separable matrix, resorts to a brute force search which is extremely inefficient. In search of an efficient decoding algorithm, we define a sufficient condition on the testing matrix, again taking inspiration from the noiseless group testing literature, called \emph{deletion-resilient disjunct matrices}. Furthermore, we provide randomized constructions of deletion-disjunct matrices with almost optimal number of rows, 
and also devise an efficient decoding algorithms. Finally, we also propose a deterministic construction similar to the Kautz-Singleton construction \cite{kautz1964nonrandom} using codes with strong distance properties. 

\subsection{Related Work}
Group testing and many of its variants are very well-studied and span a long list of excellent works. Here we survey the most relevant works on non-adaptive combinatorial group testing. 

In the noiseless setting, an exact set of $k$ out of $n$ defectives can be recovered using separable matrices \cite{du2000combinatorial}. The columns of these matrices form a binary superimposed code \cite{d2013superimposed} that ensures that the $\OR$-sum of any $k$ codewords are unique. It is well-known \cite{furedi1996onr} that for $n$ such codewords to exist, the code length has to be at least $m = \Omega(k^2 \log n  / \log k)$. A disjunct matrix is another characterization of testing matrices that enables efficient recovery of the defectives in the noiseless setting. The columns of such matrices form a $k$-cover-free family. This property guarantees that the $\OR$-sum of any $k$ columns is not covered by any other column. Various constructions of disjunct matrices have been proposed in the literature \cite{kautz1964nonrandom,porat2008explicit,inan2019optimality}. A random Bernoulli matrix with $m = O(k^2 \log n)$ rows is known to be $k$-disjunct with high probability \cite{d2014lectures}. The first explicit constructions based on Reed-Solomon codes with $m = O(k^2 \log^2 n)$ rows, was proposed by \cite{kautz1964nonrandom}. The construction is based on concatenated codes that use Reed-Solomon codes\cite{doi:10.1137/0108018} (or any MDS code \cite{1053661}) of an appropriate rate as the inner code and a disjunct matrix (in particular, the construction uses the Identity matrix which is a trivial disjunct matrix) to obtain another disjunct matrix with better parameters. This Kautz-Singleton construction paved the way for further explicit constructions of disjunct matrices and some of their variants using different inner code properties. For instance, \cite{porat2008explicit} construct disjunct matrices with $m = O(k^2 \log n)$ with codes at Gilbert-Varshamov bound, and \cite{indyk2010efficiently} construct sublinear time decodable disjunct matrices with a similar number of rows using list  decodable inner codes.

The most notable of these constructions, closely related to this work, include the results of \cite{cheraghchi2009noise}, which constructs explicit testing matrices resilient to bit-flip errors. Their results are a blend of both positive and negative notes. In particular, the author shows that it is impossible to exactly recover the defectives set in the presence of large errors. However, a near-exact recovery of the defectives can be done even with a large number of false-positive outcomes. The construction again uses a Kautz-Singleton-like construction with a slight variant of list recoverable codes as inner codes to obtain noise-resilient testing matrices with $O(d \log n)$ rows. 

In this work, we initiate the study of group testing and investigate the construction of testing matrices and exact recovery algorithms that are resilient to \emph{deletion errors}. 
To the best of our knowledge, deletion errors have not yet been explored in the context of group testing. However, some related areas of research such as channel communication \cite{haeupler2019near,guruswami2016efficiently,haeupler2017synchronization}, DNA storage \cite{cheraghchi2020coded,de2017optimal,grigorescu2022limitations}, and compressive sensing \cite{yao2021unlabeled,peng2021homomorphic} have addressed such noise models. 

Building upon the prior works on exact recovery in group testing, we aim to construct deletion-resilient superimposed codes. The $\OR$-sums of these codewords while satisfying the distinctness conditions, should also be uniquely identified from any of its $m-\Delta$ length subsequences to tolerate $\Delta$ deletions.
This latter condition has been studied in the context of insertion-deletion codes that have received significant attention in coding theory literature since the seminal works of Levenshtein, Varshamov, and Tenengolts \cite{levenshtein1966binary,varshamov1965codes}.
Since then various works have generalized the construction to obtain binary deletion codes of size $n$ and length $m$ such that $n = \Theta(2^m/m^k)$ and can tolerate up to some fixed constant $k$ bits of deletions \cite{sloane2000single,sima2020optimal,mitzenmacher2009survey}. In the high-noise regime, almost optimal constructions with $m$ approaching $O_k(\log n)$ are known \cite{schulman1999asymptotically,guruswami2016efficiently,guruswami2017deletion}. These codes are defined over a slightly larger alphabet and then converted to a binary alphabet by using standard code concatenation techniques. 
More recently, codes using synchronization strings were introduced by \cite{haeupler2017synchronization} that allow a black-box conversion of any error-correcting code to deletion code by using efficient indexing schemes. 

To design deletion-resilient superimposed codes, we will require a handle on the edit distance between $\OR$-sum of distinct sets of binary vectors. The main challenge here arises from the shifts in vectors introduced by deletions which makes deletion errors much harder to handle. We circumvent this issue by defining a suitable sufficient condition on the testing matrix that is not only satisfied by certain matrices with few rows but also facilitates efficient decoding.

\subsection{Notation} 
For any non-negative integer $n$, let $[n] := \{1, \ldots, n\}$. For any vector $v \in \R^n$, let $\supp{v} \subseteq [n]$ denote the set of indices of non-zero entries of $v$. Let $\wt(v)$ denote the number of non-zero entries in $v \in \R^n$, also called its Hamming weight. 
Let $v[i]$ denote the $i$-th entry of $v$. For $a<b$, we let $v[a:b]$ to denote the $b-a+1$-length vector $(v[a], v[a+1], \ldots, v[b])^T$ Also, for any set of indices $S \subseteq [n]$, let 
$v_{\bar{S}} \in \R^{n-|S|}$ denote the vector obtained by deleting the indices in $S$ from $v$. Let $0^n$ (and respectively $1^n$) denote a length-$n$ vector with all $0$s (respectively $1$s). 

For any two binary vectors $u,v$ of length $n$, let us denote their coordinate-wise \OR by $u \vee v$. We will interchangeably use binary vectors $v \in \{0,1\}^n$ as indicators of sets over $[n]$. Therefore, for any two binary vectors, $u,v$ of length $n$, $u \cup v$ will denote the binary indicator vector of the union of the sets indicated by $u$ and $v$, equivalent to $u \vee v$.  For binary vectors $a, b \in \{0,1\}^n$, we say $a \le b$ if $a_i \le b_i$ for every $i \in [n]$. 

For any matrix $A \in \mathbb{R}^{m \times n}$, let $A^i$ be its $i$-th row, and let $A_j$ denote its $j$-th column. The $(i,j)$-th entry of $A$ is represented as $A_j[i]$, $A^i[j]$, or $A[i,j]$ depending on context. 

\subsection{Setup}
We consider the problem of designing pooling strategies that will enable the recovery of all  $k$ defective items from a population of $n$ items using a small number of pooled tests in the presence of deletion noise. 

Let $x^* \in \{0,1\}^n$ be the indicator of ground truth defectives such that $\wt(x^*) \le k$.  A group test, denoted by $\test{a}{x^*}$, corresponds to testing if the set of elements indicated by the pooling vector $a \in \{0,1\}^n$ contains a defective or not. The outcome of such a group test is $1$ if $\supp{a} \cap \supp{x^*} \neq \emptyset$, and $0$ otherwise. Equivalently, the test outcome is given as $\test{a}{x^*} = \bigvee_{i\in \supp{x^*}} a_i$.
Let $y := \test{A}{x^*} \in \{0,1\}^m$ denote the outcomes of $m$ group tests with pooling vectors represented as the rows of the binary matrix $A \in \{0,1\}^{m \times n}$. In other words,  $y := \bigvee_{i \in \supp{x^*}} A_i$ is the $\OR$-sum of upto $k$ columns of $A$. 

Our goal is to design a testing matrix $A \in \{0,1\}^{m \times n}$ with the fewest possible rows that facilitates efficient recovery of $x^*$ given $\ti{y} \in \{0,1\}^{\ti{m}}$ ($m-\ti{m} \le \Delta$) which is a corrupted version of $y = \test{A}{x^*}$ after at most $\Delta$ arbitrary deletions. For shorthand, we use $\tilde{y} = \corr(y, \Delta)$ to denote the deletion corrupted test outcomes.
We will assume the knowledge of both $k$, and $\Delta$. An upper bound on these quantities also suffices for the purpose.

\subsection{Preliminaries}

The following matrices are used for the recovery of defectives in noiseless group testing literature.
\begin{definition}[Separable Matrix]
    A binary matrix $A \in \{0,1\}^{m \times n}$ is a $k$-separable matrix if the \OR-sum of any two subsets $S_1, S_2 \subseteq [n]$ of up to $k$ columns is distinct,
    \[
    \bigvee_{i \in S_1} A_i \neq \bigvee_{i \in S_2} A_i
    \]
\end{definition}

It is well-known that if the testing matrix $A$ is $k$ separable, then it can identify up to $k$ defective elements. 
However, to enable an efficient decoding algorithm, a $k$-disjunct matrix is used. 
\begin{definition}[Disjunct Matrix]
    A binary matrix $A \in \{0,1\}^{m \times n}$ is a $k$-disjunct matrix if for any set of $k+1$ columns, indexed by $\{i_0, i_1, \ldots, i_k\}$, $A_{i_0}$ is not contained in the $\OR$-sum of other columns, 
    \[ A_{i_0} \not \leq \bigvee_{j=1}^k A_{i_j}.\]
\end{definition}

The two above-defined types of matrices are known to be essentially equivalent. It is trivial to see that a $k$-disjunct matrix is also a $k$-separable matrix. However, a $k$-separable matrix is also a $k-1$-disjunct matrix \cite{du2000combinatorial}. 

The decoding algorithm using a disjunct matrix is presented in Algorithm~\ref{alg:disj_dec}. Here given access to $y = \test{A}{x^*}$ and the testing matrix $A$, the algorithm runs in $O(mn)$ time to recover $x^*$ exactly. 
\begin{algorithm}[h!]
\begin{flushleft}
    \textbf{Input}: Binary vector $y = \test{A}{x^*} \in \{0,1\}^{m}$
\end{flushleft}
\begin{algorithmic}[1]
\caption{\DisjDec}
\label{alg:disj_dec}
\State{$\hat{x} = [n]$}
\ForEach{test $i \in [m]$}
    \If{$y_i==0$}
        \State{$\hat{x} \leftarrow \hat x \setminus \{ j \in [n] | A_{i,j} = 1\}$ }
    \EndIf
\EndFor
\State{Return $\hat{x}$}
\end{algorithmic}
\end{algorithm}

The algorithm relies on the property of disjunct matrices that ensures that each non-defective item appears in at least one negative test.

\section{Our Results}

We start by extending the techniques from noiseless group testing. Our first result establishes a necessary condition on the testing matrix $A$ for the recovery of $x^*$ from group tests with at most $\Delta$ deletions. To state the result, we first need the following definitions: 

\begin{definition}(Deletion distance) The deletion distance $d_{del}(x,y)$ between two strings $x$ and $y$ of length $n$ is the maximum number of deletions $d$ such that any $(n-d)$-length subsequences $x'$ of $x$ and $y'$ of $y$ are distinct. 
\end{definition}
Note that deletion distance is equivalent to edit distance with only deletions. Therefore,  it follows that $d_{del}(x,y) = n - \text{lcs}(x,y) - 1$, where $\text{lcs}(x,y)$ is the length of a longest common subsequence of $x$ and $y$. For example, let $x,y\in \{0,1\}^6$, $x= [0,1,0,1,0,0]^T$ and $y = [0,0,0,1,1,0]^T$. The deletion distance $d_{del}(x,y) = 1$. The longest common subsequence is $[0,0,0,0]^T$ (or $[0,0,1,0]$ or $[0,1,1,0]^T$).

Similar to the noiseless group testing setting, we establish the necessary conditions on $A$ based on separable matrices that ensure that the test outcomes corresponding to any two distinct sets of defectives are different. 

\begin{definition}[$(k, \Delta)$-deletion separable matrix]
We call a matrix $A \in \{0,1\}^{m \times n}$ to be $(k, \Delta)$-deletion separable matrix if the \OR-sum of any two distinct subsets $S_1, S_2 \subseteq [n]$ of up to $k$ of its columns has a deletion distance of at least $\Delta$
\[
d_{del}(\bigvee_{i \in S_1} A_i, \bigvee_{i \in S_2} A_i) \ge \Delta.
\]
\end{definition}

Note that $(k, \Delta)$ - deletion separable matrices are necessary for the unique recovery of the defective set after $\Delta$ deletions. Otherwise, there will exist two distinct sets of defectives that will have the same test outcomes after $\Delta$ deletions, and therefore we cannot distinguish between them.

The deletion separable condition is also sufficient for recovery, albeit with an inefficient brute-force decoder. In particular, an upper bound on the number of tests sufficient for exact recovery can be obtained using a randomized construction of $(k, \Delta)$ - deletion separable matrix. However, the decoding algorithm must do a brute-force search over all $O(\binom{n}{k})$ sets to recover the ground truth $x^*$. Instead, we show that a trivial construction based on replicating $k$-disjunct matrices will give us a $(k, \Delta)$-deletion separable matrix with an efficient decoder (Section~\ref{sec:del-sep}). 

\begin{theorem}\label{thm:del-dist}
There exists a $(k, \Delta)$ - deletion separable matrix $A \in \{0,1\}^{m \times n}$ with $m = O(\Delta \cdot k^2\log{n})$. 
\end{theorem}

\begin{proof}[Proof of Theorem~\ref{thm:del-dist}]

Let $B \in \{0,1\}^{m' \times n}$ be a $k$-disjunct matrix with $m'$ rows. Let $A$ be the matrix obtained by repeating each row of $B$, $\Delta + 1$ times. In particular, for each $j \in [m']$,
\begin{equation}\label{eq:matrix-cons}
A_k = B_j \text{ for all } k \in [(j-1)(\Delta+1)+1, j(\Delta+1)]    
\end{equation}

Since constructions of $k$-disjunct matrices $B$ with $O(k^2 \log n)$ rows exist~\cite{porat2008explicit}, we can obtain the matrix $A$ with $m = O(\Delta k^2 \log n)$ rows. 

We now show that the matrix $A$ obtained as described above is indeed a $(k, \Delta)$-deletion separable matrix. First, note that the matrix $B$ is a $k$-separable matrix since every disjunct matrix is also separable. Therefore, for any two subsets $S_1, S_2 \subseteq [n]$ of $k$ columns of $B$, $\vee_{i \in S_1} B_i \neq \vee_{i \in S_2} B_i$. In other words, without loss of generality, there exists an index $\ell \in [m']$ such that  $(\vee_{i \in S_1} B_i)[\ell] = 0$ and $(\vee_{i \in S_2} B_i)[\ell] = 1$. Therefore, by the construction of matrix $A$, 
the $\ell$-th block of the vectors $y_1 := (\vee_{i \in S_1} A_i)$, and $y_2:=(\vee_{i \in S_2} A_i)$ will satisfy
\[
y_1[(\ell-1)(\Delta+1)+1: \ell(\Delta+1)] = 0^{\Delta+1} \qquad y_2[(\ell-1)(\Delta+1)+1: \ell(\Delta+1)] = 1^{\Delta+1}. 
\]
So on deleting any set of at most $\Delta$ entries of $y_1$ and $y_2$, there will exist at least one index from the $\ell$-th block of $y_1$ and $y_2$ that aligns and, therefore, will be unequal. So, the deletion distance between $y_1$ and $y_2$ will be at least $\Delta$. 
\end{proof}

In particular, sufficient repetitions of a disjunct matrix ensure that the outcomes corresponding to the underlying disjunct matrix can be recovered even after $\Delta$ deletions. The decoding algorithm then simply uses these outcomes to recover the ground truth. The algorithm described in Section~\ref{sec:del-sep} then runs in $O(m + mn/\Delta)$ time. 

We note that the upper bound on $m$, established in Theorem~\ref{thm:del-dist} is quite prohibitive due to the $\Delta$ factor increase compared to the noiseless setting. 
In pursuit of a better upper bound, we will define alternate sufficient conditions on $A$ for the exact recovery of $x^*$. We start by defining the following asymmetric version of deletion distance. 
\begin{definition}[Asymmetric deletion distance]
The asymmetric deletion distance $d_{adel}(x,y)$ between two strings $x$ and $y$ of length $n$ is the maximum number of deletions $d$ such that for any $(n-d)$-length subsequences $x'$ and $y'$ of $x$ and $y$ there exists a position $i$ such that $x'[i]=1$ and $y'[i]=0$.
\end{definition}

Note that the asymmetric deletion distance is stronger than the notion of deletion distance. It not only ensures that the deletion distance between $x$ and $y$ is larger than $\Delta$ but also guarantees the existence of a position $i \in [n-\Delta]$ in every substring  $x'$ and $y'$ of $x$ and $y$ respectively where $x'[i] > y'[i]$. We will call such an index $i$ the position of a $1-0$ match in $(x', y')$. For example, let $x,y\in \{0,1\}^7$, $x=[1,0,1,1,0,1,1]^T$ and $y=[0,0,0,0,1,0,0]^T$. The asymmetric deletion distance $d_{adel}(x,y) = 2$. Since for any 4-length subsequence, $x'$ and $y'$, $x'$ at least has two $1$, and $y'$ at most has one $1$. For 3-length subsequence $x'$ and $y'$, both can be $[0,0,1]^T$ and there is no $1-0$ match in $(x',y')$.

Inspired by work on noiseless group testing, we now define the following sufficient condition on $A$: $ x*$ can be recovered with significantly fewer tests, thereby bringing us closer to the known lower bound. 
\begin{definition}[$(k, \Delta)$-deletion disjunct matrix]{\label{deletion-disjunct}}
We call a matrix $A \in \{0,1\}^{m \times n}$ to be $(k, \Delta)$-deletion disjunct matrix if for any set of $k+1$ columns $A_{i_0}, A_{i_1},\ldots,A_{i_k}$ the asymmetric deletion distance between $A_{i_0}$ and $\vee_{s=1}^{k} A_{i_s}$ is at least $\Delta$, 
\[
d_{adel}(A_{i_0}, \bigvee_{s=1}^{k} A_{i_s}) \ge \Delta.
\]
\end{definition}

We remark that deletion disjunct matrices are sufficient for recovery since they are also deletion separable. This fact is formalized below. 
\begin{lemma}\label{lem:delUF-delDis}
A $(k, \Delta)$-deletion disjunct matrix is also a $(k, \Delta)$-deletion separable matrix. 
\end{lemma}
\begin{proof}
Suppose the binary matrix  $A \in \{0,1\}^{m \times n} $ is a $(k, \Delta)$-deletion disjunct matrix but is not $(k, \Delta)$-deletion separable. From the definition of a deletion separable matrix, there exist two distinct subsets of $k$ columns $S_1, S_2 \in [n]$ of $A$ such that 
$d_{del}(\vee_{i \in S_1} A_i, \vee_{j \in S_2} A_j) < \Delta$. 

Therefore, for some sets of $\Delta$ deletions indexed by rows $T_{1}, T_{2} \subseteq [m]$ , $|T_{1}| = |T_{2}| = \Delta$, 
the vectors $\vee_{i \in S_1} A_i$, and $\vee_{i \in S_2} A_i$ become indistinguishable after deleting entries from $T_1$ and $T_2$ respectively, i.e., 
$$(\bigvee_{i \in S_1} A_i)_{\overline{T_1}}  =  ( \bigvee_{j \in S_2} A_j)_{\overline{T_2}}.$$

Therefore,  $\forall i \in S_1$, $\supp{(A_i)_{\overline{T_1}}} \subseteq \supp{( \bigvee_{j \in S_2} A_j)_{\overline{T_2}}}$. In other words, $(A_i)_{\overline{T_1}} \le ( \bigvee_{j \in S_2} A_j)_{\overline{T_2}}$, and hence, there is no $1-0$ match between $(A_i)_{\overline{T}_i}$ and $(\vee_{j \in S_2} A_j )_{\overline{T}_2}$. Therefore, $d_{adel}(A_i,\vee_{j \in S_2} A_j )\le\Delta$, 
which gives a contradiction to the assumption that $A$ is a $(k, \Delta)$-deletion disjunct matrix. 
\end{proof}

The equivalence of the two matrices enables us to establish a lower bound on the number of tests. 
\begin{theorem}\label{thm:del-dist-new}
Let $A \in \{0,1 \}^{m \times n}$ be a $(k, \Delta)$-deletion disjunct matrix, then 
\begin{center}
    $m =\Omega(k^2 \log n / \log k + k\Delta)$. 
\end{center}
\end{theorem}
\begin{proof}
Consider any fixed set of $ k+1 $ columns of $A$ indexed by $S = \{i_0, i_1, \ldots, i_k\}$. 
By the definition of a deletion disjunct matrix, there would exist a subset of row indices, $Q \subseteq [m]$ of $|Q| \geq \Delta + 1$,   such that for every $j \in Q$, $A_{i_0}[j] = 1$ and $(\vee_{s=1}^{k} A_{i_s})[j] = 0$. Since, $d_{adel}(A_{i_0}, \bigvee_{s=1}^{k} A_{i_s}) < \Delta$ if this were not the case. 
A similar argument applies to other combinations of columns indexed by $S$ that also satisfy $d_{adel}(A_{i_u}, \bigvee_{i_s\in S\setminus\{i_u\}}^{k}) \geq \Delta$ (role of $i_0$ and $i_u$ is swapped). 
Therefore, each column $A_{i_u}$, $i_u \in S$ has at least $\Delta + 1$ distinct non-zero entries which implies that $m \geq (k+1)(\Delta+1)$.

Additionally, note that since a deletion disjunct matrix is also deletion separable. Therefore, deletion of any $\Delta$ arbitrary rows still results in a $k$-separable matrix. From the results in \cite{furedi1996onr}, this gives us another lower bound: $m - \Delta \geq \Omega\left(\frac{k^2 \log n}{\log k}\right)$. Therefore, from Lemma \ref{lem:delUF-delDis} we obtain $m \geq \max \left\{ (k+1)(\Delta+1), \Omega\left(\frac{k^2 \log n}{\log k} + \Delta\right) \right\}$. 

Thus, combining the two observations, it follows that  $ m = \Omega\left(\frac{k^2 \log n}{\log k} + k\Delta\right)$ for both deletion disjunct and deletion separable matrices.
\end{proof}

Furthermore, we show that a random Bernoulli matrix with $m = O(k^2 \log n + \Delta k)$ rows is deletion disjunct with high probability.
Thereby proving the existence of deletion disjunct matrices with an almost optimal number of rows. 
\begin{theorem}\label{thm:random-del-disj}
    There exists a $(k, \Delta)$-deletion disjunct matrix with $m = O(k^2 \log n + \Delta k)$ rows.
\end{theorem}

We not only get closer to the established lower bound but also obtain an efficient decoding algorithm to recover $x^*$ given any deletion disjunct matrix. The proof of Theorem~\ref{thm:del-disj-dec} proposes a non-trivial generalization of Algorithm~\ref{alg:disj_dec} that is robust to deletions, and runs in $O(mn)$-time. 

\begin{theorem}\label{thm:del-disj-dec}
Let $A \in \{0,1\}^{m\times n}$ be a $(k,\Delta)$-deletion disjunct matrix. There exists a decoding algorithm that runs in time $O(mn)$ and recovers $x^*$ given $\ti{y} = \corr(\test{A}{x^*}, \Delta)$.
\end{theorem}

The proof of Theorem~\ref{thm:del-disj-dec} is presented in Section~\ref{sec:del-disj}. 
We first present a simple algorithm that does better than a naive brute force algorithm and runs in time $O(nm^{\Delta+1})$ (see Algorithm~\ref{alg:del-disj}). While this algorithm can efficiently handle a small number of deletions ($\Delta = \Omega(1)$), the exponential dependence on $\Delta$ makes it prohibitive for larger corruptions. 
We, therefore, optimize a crucial subroutine using Algorithm~\ref{alg:coverage} that reduces the runtime to $O(mn)$. 

While $O(mn)$-time is efficient, it is still prohibitive for certain applications. Similar to the works of \cite{indyk2010efficiently, cai2013grotesque, SAFFRON, guruswami2023noise} for bit-flip errors, we propose a randomized construction of a deletion-disjunct matrix that enables super-efficient (sub-linear time) decoding that runs in $\textsc{poly}(k, \log n)$ time. 

\begin{theorem}\label{thm:singleton}
    There exists a testing matrix $A \in \{0,1\}^{m \times n}$ with $m = O(k^2\log^2 n)$ rows that enables the recovery of $k$ defectives with probability at least $1-1/k$ even after $\Delta = O(\log n)$ deletions in $\textsc{poly}(k, \log n)$ time. 
\end{theorem}

Finally, for a deterministic construction of $(k, \Delta)$-deletion disjunct matrices, we adopt a technique similar to the Kautz-Singleton construction \cite{kautz1964nonrandom} using codes with strong distance properties. 

We say a set of $q^K$ vectors $C \subseteq \F_q^N$ is an $(N, K, D)_q$ code if every pair of vectors in $C$ have a Hamming distance of at least $D$. To define the stronger distance property, fix an arbitrary ordering of elements of $\mathbb{F}_q$ as $\alpha_1 < \alpha_2 < \ldots < \alpha_q$. Let $\psi$ be map $\psi: \mathbb{F}_q \rightarrow \mathbb{Z}$ defined as $\psi(\alpha_i) = i$.  

\begin{definition}[$(\Delta, d)-\ell_\infty$ distance property]
    Let $C \subseteq \mathbb{F}_q^N$ be an $(N,K,D)_q$ code. 
    For any $d \le D$, and $\Delta < d$, $C$ is said to have $(\Delta,d)-l_\infty$ distance property if for any two distinct codewords, $c_1,c_2\in C$, there exists a set of at least $d$ indices $S\subseteq [N]$, $|S| \geq d$, such that 
    \[|\psi(c_1[i]) - \psi(c_2[i]) | \geq \Delta ~\forall~ i\in S.\] 
\end{definition}

We show that $(N, K, D)_q$ codes with $(\Delta, d)-\ell_\infty$ distance property give $(k, \Delta)$ disjunct matrices.

\begin{theorem}\label{thm:kddel}
    Let $C\subseteq \mathbb{F}_q^N$ be an $(N,K,D)_q$ code with $(\Delta,d)-\ell_\infty$ property then there exists a $(k,\Delta)$-deletion disjunct matrix for any $k \leq \frac{N}{N-d+\Delta}$. 
\end{theorem}

We also present a simple construction of codes with $(\Delta, d)-\ell_\infty$ distance property. 
\begin{theorem}\label{thm:code}
    Let $q$ be prime power, $\Delta < q$, and $q'$ be a prime power larger than $\Delta q + 1$. Then any $(N,K,D)_q$ code can be embedded into $\mathbb{F}_{q'}$ to obtain a $(N,K,D)_{q'}$ code  with $(\Delta, D)-\ell_\infty$ distance property.
\end{theorem}

However, the construction in Theorem~\ref{thm:kddel} instantiated using MDS codes does not provide any improvement over the trivial repetition construction shown in Theorem~\ref{thm:del-dist}. We therefore, leave the search for $(N, K, D)_q$ codes with $(\Delta, d)-\ell_\infty$ distance property with better parameters an open problem. 

The summary of our results can be found in the following Table~\ref{tab:results}.
\begin{table}[htbp]
\centering           

\begin{tabularx}{1.0\textwidth} 
{ 
  | >{\raggedright\arraybackslash}p{2.2cm}    
  | >{\raggedright\arraybackslash}X 
  | >{\raggedright\arraybackslash}p{4.3cm} | } 
 \hline
 \textbf{Theorem}  & \textbf{Description} & \textbf{Result} \\
 \hline
  Theorem \ref{thm:del-dist}  & Upper bound for deletion separable matrix  & $m = O(\Delta \cdot k^2\log{n})$  \\
  \hline
  Theorem \ref{thm:del-dist-new} & Lower bound for deletion separable matrix &  $m =\Omega(k^2 \frac{\log n}{\log k}  + k\Delta)$   \\
  \hline
  Theorem \ref{thm:random-del-disj}, \ref{thm:del-disj-dec} & Random construction/ Decoding algorithm for deletion disjunct matrix & $m = O(k^2 \log n + \Delta k)$, decoding time: $O(mn)$\\
  \hline
  Theorem \ref{thm:singleton} & Running time/Number of test for Sub-linear algorithm of deletion disjunct matrix & $m = O(k^2 \log k \log^2 n)$, decoding time: $\textsc{poly}(k, \log n)$\\
  \hline 
  Theorem \ref{thm:kddel},\ref{thm:code} & Number of test for deterministic deletion disjunct matrix & $ m = O(\Delta \cdot k^2 \log n)$\\
\hline
\end{tabularx}\label{tab:results}
\caption{Summary of Theorems and Results}  
\end{table}


\textbf{Organization: }
The randomized construction of deletion disjunct matrices (Theorem~\ref{thm:random-del-disj}), and the associated decoding algorithm (Theorem~\ref{thm:del-disj-dec}) are presented in Section~\ref{sec:del-disj}. 
Section~\ref{sec:disjunct-sublinear} contains the construction of deletion disjunct matrices with sub-linear decoding time algorithm (Theorem~\ref{thm:singleton}). Finally, the deterministic construction (Theorem~\ref{thm:kddel}, Theorem~\ref{thm:code}) is presented in Section~\ref{sec:disjunct-deterministic}. 


Due to the simplicity of the result, we present the trivial construction of deletion separable matrices (Theorem~\ref{thm:del-dist}) and its associated efficient decoder in Section~\ref{sec:del-sep}.

\section{Construction of $(k, \Delta)$ - Deletion Separable Matrix \& Decoding}\label{sec:del-sep}
We now present our construction of the deletion separable matrices. For this, we use disjunct matrices that have been classically used in noiseless group testing. 
Specifically, we replicate a disjunct matrix multiple times to make the outcomes resilient to deletion noise. 

\subsection{Decoding Algorithm}
Let $\tilde{y} = \corr(y, \Delta) \in \{0,1\}^{m-\Delta}$ be the output of the group tests $y = \test{A}{x^*}$ after $\Delta$ deletions, where $A$ is the deletion separable matrix constructed above in Theorem~\ref{thm:del-dist}.

The decoding algorithm to recover the underlying sparse vector $x^*$ from corrupted test outcomes is presented in Algorithm~\ref{alg:cap}. The algorithm proceeds in two steps. In the first step, using a greedy approach, we reconstruct the actual test outcomes, $y$, from the deletion corrupted one, $\ti{y}$ (Algorithm~\ref{alg:greedycomp}). This recovery is possible due to the specific repetition structure of the deletion separable matrix $A$ designed in the previous section. Once we obtain the true test outcomes, we can recover $x^*$ using the decoding algorithm for the disjunct matrix $B$, presented in Algorithm~\ref{alg:disj_dec} for completeness.

The procedure $\GreedyComp$, presented in Algorithm~\ref{alg:greedycomp}, corrects the vector $\tilde{y}$ to recover the actual test outcome vector $y$. A run of symbol $b$ (either $0$ or $1$) in $\tilde{y}$, of length $\ell$, is a substring of length $\ell$ of $\tilde{y}$ that contains only $b$. The procedure extends each run of a particular symbol $b \in \{0,1\}$ of length $\ell$ to a complete multiple of $(\Delta+1)$. We now show that the vector returned by $\GreedyComp$ is $y$ if there are at most $\Delta$ deletions.

\begin{algorithm}[h!]
\begin{flushleft}
    \textbf{Input}: Binary vector $\tilde{y} \in \{0,1\}^{m-\Delta}$, Block size $\Delta+1$
\end{flushleft}
\begin{algorithmic}[1]
\caption{\GreedyComp}\label{alg:greedycomp}
\State{$\hat{y} = ()$}
\ForEach{run of symbol $b \in \{0,1\}$ of length $\alpha(\Delta+1)+\beta$ in $\tilde{y}$, ($\alpha, \beta \in \Z$, $0 \le \beta \le \Delta$)}
    \If{$\beta==0$}
        \State{Append a run of symbol $b$ of length $\alpha(\Delta+1)$ to $\hat{y}$}
    \Else
        \State{Append a run of symbol $b$ of length $(\alpha+1)(\Delta+1)$ to $\hat{y}$}
    \EndIf
\EndFor
\State{Return $\hat{y}$}
\end{algorithmic}
\end{algorithm}

\begin{lemma}\label{lem:greedy-correct}
Let $\hat{y}$ be the output of $\GreedyComp$ on input $\tilde{y}=\corr(y, \Delta)$, then $\hat{y} = y$. 
\end{lemma}
\begin{proof}
    We prove this lemma by induction on the blocks of identical rows in $A$ that have one-one correspondence with individual rows in $B$,
    and $m' = m/(\Delta+1))$. 
    
    \textbf{Base case: $m' =1$:} In this case, note that length of $y$ is $\Delta+1$, and $\tilde{y}$, WLOG is of length $1$. We now show that $\hat{y}[i] = y[i]$ for each $i \in [1,\Delta+1]$. Since the first row of $B$ is repeated $\Delta+1$ times in $A$, the first $\Delta+1$ test outcomes in $y$ are identical. Therefore, even after $\Delta$ deletions, one of the true outcomes with the first row of $B$ is preserved in $\tilde y$. The algorithm, therefore, reconstructs and returns $\hat y = y$. 
    
    \textbf{Induction Hypothesis: } Let us assume the correctness of the algorithm, $\GreedyComp$ for up to the $t$-th block of $A$, i.e., $m' \le t$. 
    
    \textbf{Induction step:} 
    Consider the true test output vector $y$ until the $t+1$-th block of $A$. Let $\tilde y$ be a vector of length $\ell := (t+1)(\Delta+1) - \Delta$ obtained by deleting arbitrary $\Delta$ entries of $y$. 
    Let $\Delta_1$ denote the number of deletions in the first $t(\Delta+1)$ entries of $y$. Define $p := t(\Delta+1) - \Delta_1$. Then, $\tilde{y}_1 := \tilde y [1:p]$ denotes the test outcomes from the first $t$ blocks of $A$, and $\tilde{y}_2 := \tilde y [p+1:\ell]$ are the outcomes from the $t+1$-th block. 
    
    By induction hypothesis, $\hat{y}_1 = \GreedyComp(\tilde{y_1})$ will be correctly computed and we have $\hat{y}_1 = y[1: t(\Delta+1)]$. 
    Also, since all rows in the $t+1$-th block of $A$ are identical, $\tilde y [p+1:\ell] = 0^{\ell-p}$(or, $1^{\ell-p}$). $\GreedyComp$ will extend this run of $\ell - p$ $0$s (or 1s) to a length of $\Delta+1$ elements. So, $\hat{y}_2 = y[t(\Delta+1)+1 : (t+1)(\Delta+1)]$. Combining the two parts, the correctness of $\GreedyComp$ follows. 
\end{proof}

Once the actual test outcome vector $y = \test{A}{x^*}$ is recovered, we can reconstruct $\test{B}{x^*}$ by removing the redundancy. The ground truth defective status $x^*$ can then be obtained using any decoding algorithm for disjunct matrices. These steps are listed in Algorithm~\ref{alg:cap}. 

\begin{algorithm}[h!]
\begin{flushleft}
\textbf{Input}: $k$-disjunct matrix $B \in \{0,1\}^{m' \times n}$, \\
Matrix $A \in \{0,1\}^{m \times n}$ generated from $B$ as described in Equation~\ref{eq:matrix-cons}\\
Corrupted output $\tilde{\f{y}} := \corr(\f{y}, \Delta)$, where $\f{y} = \test{A}{x^*}$. \\ 
\end{flushleft}
\begin{algorithmic}[1]
\caption{Algorithm to recover the set of defectives}\label{alg:cap}
\State{Let $\hat{y} := \GreedyComp(y', \Delta)$}
\Comment{Fix deletions greedily to recover $\test{A}{x^*}$}
\State{Let $\mathbf{y}' = 0^{m'}$}
\ForEach{$j \in [m/(\Delta + 1)]$}
        \State{$\mathbf{y}'[j] = \hat{y}[j(\Delta + 1)]$}
\EndFor 
\Comment{Construct $\test{B}{x^*}$}
\State{$\hat{x} := \DisjDec (B, y')$} \Comment{See Algorithm~\ref{alg:disj_dec}}
\State{return $\hat{\mathbf{x}}$}
\end{algorithmic}
\end{algorithm}

\begin{theorem}
    Let $A \in \{0,1\}^{m \times n}$ be the deletion separable matrix constructed from a $k$-disjunct matrix $B$ as described in Theorem~\ref{thm:del-dist}. Given $\tilde{y} = \corr(\test{A}{x^*}, \Delta)$, Algorithm~\ref{alg:cap} recovers $x^*$ exactly in time $O(mn/\Delta)$.
\end{theorem}

\begin{proof}
    The correctness of Algorithm~\ref{alg:cap} follows from the correctness of Lemma~\ref{lem:greedy-correct} and the decoding algorithm for noiseless group testing. The latter follows from construction since $B$ is a $k$-disjunct matrix. 

    For the run-time analysis, note that Algorithm~\ref{alg:greedycomp} runs in time $O(m)$, and the column matching algorithm takes $O(m'n)$ time. Therefore, Algorithm~\ref{alg:cap} runs in time $O(m + mn/\Delta)$. 
\end{proof}

\section{Deletion disjunct matrix \& decoding}\label{sec:del-disj}
In this section, we will consider random designs to obtain a $(k, \Delta)$-deletion disjunct matrix with high probability. 
\begin{proof}[Proof of Theorem~\ref{thm:random-del-disj}]
Let $A \in \{0,1\}^{m \times n}$ be a random Bernoulli matrix where each entry is set to $1$ independently with some probability $p > 0$. We show that for an appropriate value of $p$, and large enough $m$, $A$ will be a $(k, \Delta)$ - deletion disjunct matrix with high probability. 

Recall that for any two vectors $v_1, v_2 \in \{0,1\}^{m-\Delta}$, we say that a $1-0$ match occurs if there exists an index $i$ such that $v_1[i] = 1$, but $v_2[i] = 0$. 
For a fixed set of $k+1$ columns indexed by $T \subset [n]$, $|T| =k$, and $j \in [n] \setminus T$, and any two sets of $\Delta$ deletions indexed by $S_1, S_2 \subseteq [m]$, $|S_1| = |S_2|= \Delta$, define a bad event $E(T, j, S_1, S_2)$ to occur if there is no $1-0$ match between $(A_j)_{\overline{S_2}}$ and $(\vee_{i \in T} A_i)_{\overline{S_1}}$. 
Note that $A$ is not $(k, \Delta)$-deletion disjunct matrix if $E(T, j, S_1, S_2)$ occurs for any set $T, j, S_1, S_2$ since then, $d_{adel}(A_j, \vee_{i \in T} A_i) < \Delta$. 
We will now bound the probability of such bad events from occurring. 

Observe from the independence of the matrix entries that for a fixed row, $w \in [m-\Delta]$, the probability that a $1-0$ match does not occur is given by $(1-(1-p)^k\cdot p)$. Therefore, for a fixed set of $T, j, S_1, S_2$, the event $E(T, j, S_1, S_2)$ occurs with probability $(1-(1-p)^k\cdot p)^{m-\Delta}$. 

Taking a union bound over all possible $\binom{n}{k} (n-k)$ choices of $T, j$ and $\binom{m}{\Delta}^2$ choices of $S_1, S_2$, we get that the failure probability is 
\[(1-(1-p)^k\cdot p)^{m-\Delta}{m \choose \Delta}^2 \binom{n}{k} (n-k).\]
Setting $p=1/k$, 
we get that the probability that $A$ is not a $(k, \Delta)$-deletion disjunct matrix is
\begin{align*}
    p_{\text{fail}} &:=(1-(1-p)^k\cdot p)^{m-\Delta}{m \choose \Delta}^2 \binom{n}{k} (n-k) \\
    &\le \left(1-(1-\frac1k)^k \frac 1k\right)^{m-\Delta} \left(\frac{em}{\Delta}\right)^{2\Delta} \left(\frac{en}{k}\right)^k (n-k) \quad (\text{Since, }\binom{m}{\Delta} \le (em/\Delta)^\Delta)\\
    &\le exp\left(-\frac{(m-\Delta)}{10k} + 2\Delta \log \frac{em}{\Delta} + k \log \frac{en}{k} + \log (n-k) \right), 
\end{align*}
where the last inequality follows from using the standard approximations for $e^{-2x} \le (1-x) \le e^{-x}$ for any $x \in (0,1/2)$. Therefore, if we set $m = \tilde{O}\left(k^2 \cdot \log{n} + \Delta \cdot k\right)$, where $ \tilde{O}(\cdot)$ hides $\log \log n$ and $\log k$ terms will ensure that with probability $1-1/n$, $A$ will be $(k,\Delta)$-deleltion disjunct. 
\end{proof}

\subsection{Decoding Algorithm}
In this section, we will give a decoding algorithm to recover $x^*$ from $\ti{y} = \corr(y, \Delta)$ using a $(k, \Delta)$-deletion disjunct matrix that runs in time $O(mn)$. We first present a simple algorithm that has a slightly degraded running time, which we improve later. 

Consider Algorithm~\ref{alg:del-disj}. We show that Algorithm~\ref{alg:del-disj} correctly decodes $x^*$ from $\tilde{y}$, but runs in time $O(nm^{\Delta+1})$. 
\begin{algorithm}[h!]
\begin{flushleft}
\textbf{Input}: $\mathbf{A} \in \{0,1\}^{m \times n}$ --  $(k, \Delta)$-deletion disjunct matrix \\
$\ti{\mathbf{y}} = \corr(\test{A}{x^*}, \Delta) \in \{0,1\}^{(m-\Delta)}$\\
\end{flushleft}
\begin{algorithmic}[1]
\caption{Decode with Deletion Disjunct Matrix}\label{alg:del-disj}
\State {Let $\hat{x} = 1^n$}
\ForEach{$j \in [n]$} 
    \If {$\forall T \subseteq [m]$ with $|T|=\Delta - 1$, $\exists i \in [m] \setminus T$  such that $(A_j)_{\overline{T}} [i] = 1$ and $\ti{y}_i = 0$}
        \State {Set $\hat{x}[j] = 0$}
    \EndIf 
\EndFor
\State {Return $\hat{x}$}
\end{algorithmic}
\end{algorithm}

\begin{lemma}\label{lem:alg1correct}
    Let $A \in \{0,1\}^{m\times n}$ be a $(k,\Delta)$-deletion disjunct matrix. Then, Algorithm~\ref{alg:del-disj} runs in time $O(nm^{\Delta+1})$ and recovers $x^*$ given $\ti{y} = \corr(\test{A}{x^*}, \Delta)$.
\end{lemma}

\begin{proof}[Proof of Lemma~\ref{lem:alg1correct}]
The proof of Lemma~\ref{lem:alg1correct} follows form the deletion disjunct property and can be considered as a generalization of Algorithm~\ref{alg:disj_dec}.  
The algorithm relies on the observation that for $S = \supp{x^*}$, and any $j \in [n]\setminus S$, $d_{adel}(A_j, \vee_{i\in S} A_i) \ge \Delta$. The observation follows directly from the definition of $(k,\Delta)$-deletion disjunct matrix. 
Therefore, the algorithm iteratively checks and rejects those columns of $A$ whose asymmetric deletion distance from $y (=\test{A}{x^*} = \vee_{i\in S} A_i)$ is larger than $\Delta$. For the check, the algorithm does a brute-force search over all possible deletions $T$ of size $\Delta$ to find one set of deletions that has no $1-0$ match with $\ti{y}$. The non-existence of $1-0$ some $T$-deletions gives a certificate that $ j \in \supp{x}$.


To prove the correctness of the algorithm, we show both directions -- 1) If $i \in \supp{x^*}$, then $\hat{x}_i = 1$, and 2) If $i \notin \supp{x^*}$, then $\hat{x}_i = 0$

\paragraph{Case 1: ($i \in \supp{x^*}$):} 
For each $i\in \supp{x^*}$, we show that there always exists some set of deletions indexed by $T$ such that there is no $1-0$ match between $(A_i)_{\overline{T}}$, and $\ti{y}$. Therefore, the condition within the \textsc{If} statement (Line 3 of Algorithm~\ref{alg:del-disj}) is not satisfied for all subsets of $\Delta$ deletions, and $\hat{x}_i$ remains $1$ as initialized. 

Let $T^*$ denote the true set of $\Delta$ deletions that produce $\ti{y}$ from $y$, i.e., $\ti{y} = y_{\overline{T}^*}$. Since $y = \vee_{j \in \supp{x^*}} A_j$, we know that there is no $1-0$ match between $A_i$ and $y$, for any $i \in \supp{x^*}$. Therefore, for the true set of deletions, $T^*$, there will not exist any $1-0$ match between $(A_i)_{\overline{T}^*}$ and $y_{\overline{T}^*} = \tilde{y}$.

\paragraph{Case 2: ($i \notin \supp{x^*}$):} The second part of the proof follows from the definition of a $(k,\Delta)$ deletion disjunct matrix. Since $A$ is $(k,\Delta)$-deletion disjunct matrix, for every $i \notin \supp{x^*}$, we know that $d_{adel}(A_i, y) \ge \Delta$. Therefore, by definition, for any set of $\Delta$ deletions $T_1, T_2$, there will be a $1-0$ match between $(A_i)_{\overline{T}_1}$ and $y_{\overline{T}_2}$. In particular, for $T_2 = T^*$, the true set of deletions that produce $\ti{y}$, there will be a $1-0$ match between $(A_i)_{\overline{T}_1}$ and $\ti{y}$ for every $T_1$. Therefore, the condition in $\textsc{If}$ statement is satisfied, and $\hat{x}_i$ is set to $0$ for every $i \notin \supp{x^*}$.

For the run-time analysis, note that the condition in the $\textsc{If}$ statement checks every possible set of $\Delta$ deletions for the existence of a $1-0$ match with $\ti{y}$. This takes $O(\binom{m}{\Delta} (m-\Delta))$. The Algorithm further iterates over all $j \in [n]$. Therefore, Algorithm~\ref{alg:del-disj} runs in time $O(nm^{\Delta+1})$.
\end{proof}

We now address the main bottleneck in Algorithm~\ref{alg:del-disj} that stems from the search for the certificate for containment in Line 3 of the algorithm. We provide an alternate mechanism to speed up this search process in time $O(m)$ instead of $O(m^{\Delta+1})$.  

\begin{algorithm}[h!]
\caption{\CheckCov}
\label{alg:coverage}
\begin{algorithmic}[1]
\Require Binary vectors $\mathbf{y}$, $\mathbf{z}$, and integer $t$
\Ensure \texttt{True} if there exists a vector $\mathbf{x}$ (obtained from $\mathbf{z}$ by deleting at most $t$ symbols) such that $\mathbf{x}$ is covered by $\mathbf{y}$, otherwise \texttt{False}.
\State $i \gets 0$
\State $j \gets 0$
\While{$i < \text{length}(\mathbf{y})$ and $j < \text{length}(\mathbf{z})$}
    \If{$\mathbf{y}[i] \geq \mathbf{z}[j]$}
        \State $i \gets i + 1$
        \State $j \gets j + 1$
    \Else
        \State $j \gets j + 1$
        \State $t \gets t - 1$
    \EndIf
    \If{$t < 0$}
        \State \Return \texttt{False}
    \EndIf
\EndWhile
\If{$i == \text{length}(\mathbf{y})$}
    \State \Return \texttt{True}
\EndIf
\State \Return \texttt{False}
\end{algorithmic}
\end{algorithm}

\begin{lemma}\label{lem:alg2correct}
    Given two binary vectors $y \in \{0,1\}^{m-t}$, $z \in \{0,1\}^{m}$, Algorithm~\ref{alg:coverage} decides in time $O(m)$ if there exist some set of $t$ deletions in $z$ such that $x \le y$, where $x \in \{0,1\}^{m-t}$ is the substring of $z$ obtained after deletions. 
\end{lemma}

\begin{proof}[Proof of Lemma~\ref{lem:alg2correct}]
The algorithm, at a high level, iteratively looks for a violation for coverage, i.e., an index $i$ such that $z_i \ge y_i$. Every time it encounters a violation, the algorithm will delete that index from $z$. Observe that only the $1$ entries are deleted from $z$. The algorithm returns a $\texttt{False}$ if it finds more than $t$ violations. 

To prove the correctness of Algorithm ~\ref{alg:coverage}, we have to show that the algorithm accepts if and only if there exists a $x$, obtained from $z$ after at most $t$ deletions that can be covered by $y$.

One side of the statement follows trivially. Note that if the algorithm accepts, then there exists an $x$ that can be covered by $y$ since the algorithm constructs one such vector from $z$ with at most $t$ deletions.

The other side of the proof is a little more intricate. We must show that if there exists a $x$ that can be covered by $y$, then the algorithm returns \texttt{True}.  

From the hypothesis, we know that there exists an $x$ that can be obtained from $z$ after $t$ deletions such that $y \ge x$. Let $s_1, s_2, \ldots, s_t$ be the indices of deletions in $z$ that produces $x$.  Let $r_1,\ldots,r_t$ be the indices of deletions in $z$ made by the algorithm to obtain $x'$. Let us denote by $z^{(i)}$ the snapshot of $z$ after deletions $s_1, \ldots s_i$. Analogously, let $z'^{(i)}$ be the state of $z$ after deletions $r_1, \ldots, r_i$. 

Let $\ell \in [t]$ be the first position of mismatch in the deletion sets, i.e., $r_\ell \neq s_\ell$, and $r_i = s_i$ for all $i < \ell$. Therefore, for the first $\ell-1$ deletions, both the snapshot of $z$, $z^{(i)}$ and $z'^{(i)}$ looks exactly the same for all $i < \ell$.  Now consider the following two cases:

1) $r_\ell < s_\ell$:  Since Algorithm~$\ref{alg:coverage}$ deletes an element only when it detects a violation (i.e., $z'_{r_\ell} = 1$ and corresponding $y_{r_\ell - \ell + 1} = 0$), the deletions $s_{\ell}, \ldots, s_{t}$ will fail to remove this violation if $s_\ell > r_\ell$. Therefore, $x = z^{(t)}$ is not covered by $y$, thereby contradicting our hypothesis. 

2) $r_\ell > s_\ell$: 
Consider the snapshots $z^{(\ell-1)}$ and $z'^{(\ell-1)}$ of $z$ that look exactly the same before the $\ell$-th deletion. Since the algorithm does not find any violations before $z_{r_\ell}$, the sub-vector $z[s_{\ell}: r_\ell]$ satisfies $z[s_{\ell},r_{\ell}] \leq y[s_{\ell}-\ell+1 , r_{\ell} - \ell + 1]$. Therefore, all the deletions between $s_{\ell}$ and $r_{\ell}$ in $z^{(\ell-1)}$ can be moved after $r_{\ell}$ without altering any violations. This gives us an equivalent $z$ with $r_\ell = s_\ell$.

The proof can then be completed inductively.
\end{proof}

Equipped with both Lemma~\ref{lem:alg1correct} and Lemma~\ref{lem:alg2correct}, we now prove Theorem~\ref{thm:del-disj-dec}. 

\begin{proof}[Proof of Theorem~\ref{thm:del-disj-dec}]
The proof of Theorem~\ref{thm:del-disj-dec} follows from the correctness of Algorithm~\ref{alg:del-disj} and Algorithm~\ref{alg:coverage} shown in Lemma~\ref{lem:alg1correct} and Lemma~\ref{lem:alg2correct} respectively. 

Observe that if the asymmetric deletion distance between two vectors $x, y$ is less than $\Delta$, then there exist some $\Delta$ deletions in both $x$ and $y$ such that the substrings $x'$ and $y'$ of $x$ and $y$ respectively satisfy $x' \le y'$. This condition is equivalent to the non-existence of a $1-0$ match between $(x',y')$. Algorithm~\ref{alg:coverage} checks for this exact condition. 

In particular, for a fixed $y' =\tilde{y}$, and $t=\Delta$, Algorithm~\ref{alg:coverage} can decide efficiently if there are some $\Delta$ deletions in a column $A_j$ of $A$ such that $A_j' \le \tilde{y}$, where $A_j'$ is the vector obtained by deleting some $\Delta$ entries of $A$. Thereby, efficiently searching for a certificate for $i \in \supp{x^*}$. 

Algorithm~\ref{alg:del-disj} can therefore be equivalently described as follows. 
\begin{algorithm}[h!]
\begin{flushleft}
\textbf{Input}: $\mathbf{A} \in \{0,1\}^{m \times n}$ --  $(k, \Delta)$-deletion disjunct matrix \\
$\ti{\mathbf{y}} = \corr(\test{A}{x^*}, \Delta) \in \{0,1\}^{(m-\Delta)}$\\
\end{flushleft}
\begin{algorithmic}[1]
\caption{Decode with Deletion Disjunct Matrix}\label{alg:del-disj-new}
\State {Let $\hat{x} = 1^n$}
\ForEach{$j \in [n]$} 
    \If{$\CheckCov(\tilde{y},A_j,\Delta)==\texttt{False}$}
        \State {Set $\hat{x}[j] = 0$}
    \EndIf 
\EndFor
\State {Return $\hat{x}$}
\end{algorithmic}
\end{algorithm}
Since Algorithm~\ref{alg:coverage} runs in time $O(m)$, we can guarantee that Algorithm~\ref{alg:del-disj-new} terminates in $O(mn)$ time. 
\end{proof}

\section{Deletion Disjunct Matrix with Sublinear Time Decoding}\label{sec:disjunct-sublinear}

The construction of testing matrix $A$ is slightly similar to the one used in \cite{SAFFRON, guruswami2023noise} that are known to be robust to random bit-flip noise. Using certain coding theoretic tools, we provide constructions of testing matrices that are robust to adversarial deletions as well. We remark that similar construction can also be used to achieve robustness against a small number of adversarial bit-flips and random deletions.

    First, we will review the construction of \cite{SAFFRON} to lay the foundations for the current one. The construction involves the following two main components: 
    \begin{enumerate}
        \item A random bipartite graph  $G = (L, R, E)$ with $n$ left vertices, and $M$ right vertices that dictates the tests in which each item participates. 
        \item A signature matrix $U \in \{0,1\}^{h \times n}$ that assigns a unique binary signature of length $h$ to each item in the population. 
    \end{enumerate}

    The construction assigns each left vertex $v_i \in L$ with a signature $U_i \in \{0,1\}^h$. Each right vertex $z_j \in R$ constructs a $h \times n$ binary matrix $Z^{(j)}$ by setting its $i$-th column as
    $$Z^{(j)}_i = \begin{cases}
        U_i \text{ if } (i,j) \in E\\
        \overline{0} \text{ otherwise. }
    \end{cases} .$$ 

    The testing matrix $A \in \{0,1\}^{Mh \times n}$ is obtained by stacking each of the $Z^{(j)}$ submatrices. Therefore, 
    $A = [Z^{(1) T}, \ldots, Z^{(M) T}]^T$. 

Consider the following example of a testing matrix $A$ constructed for $3$-items using a bipartite graph with adjacency matrix given by $A_{\mathcal{G}}$ and signature matrix $U$. 
\[
A_{\mathcal{G}} = \begin{bmatrix}
1 & 1 & 1 \\
0 & 1 & 1 
\end{bmatrix} 
\qquad
U = \begin{bmatrix}
 0 & 1 & 1\\
 1 & 0 & 1\\
\end{bmatrix}
\qquad
A = \begin{bmatrix}
 0 & 1 & 1\\
 1 & 0 & 1\\
\hline
0 & 1 & 1 \\
0 & 0 & 1 \\
\end{bmatrix}
\]
The corresponding graph and test outcomes are shown in Figure~\ref{fig:testing}.
\begin{figure}[h!]
    \centering
    \includegraphics[width=0.3\linewidth]{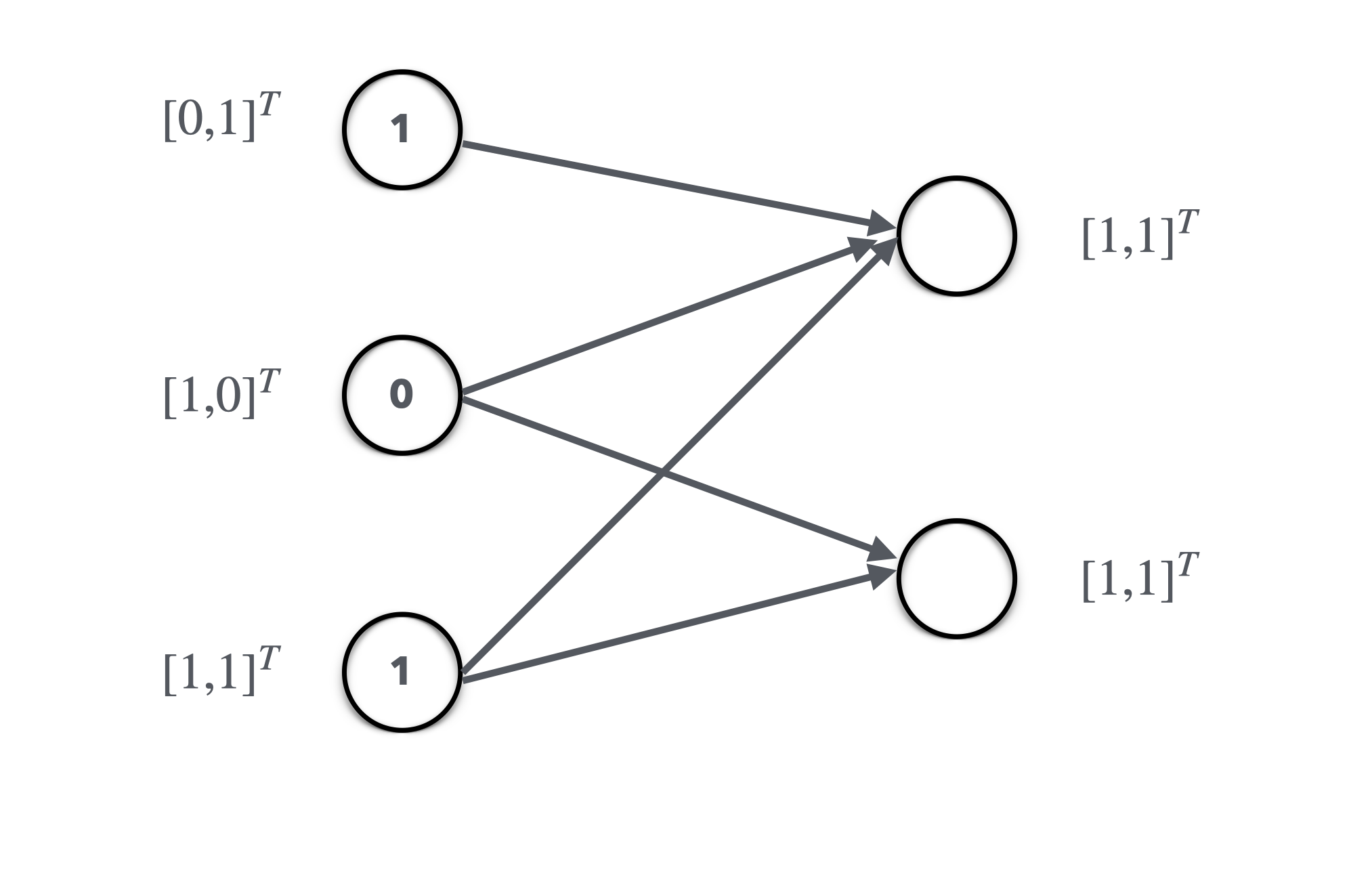}
    \caption{Testing Matrix Construction}
    \label{fig:testing}
\end{figure}
%

Note that the output of testing (in the noiseless setting) $\test{A}{x^*}$ can be partitioned into $M$ blocks, $y = (y_1^T, \ldots, y_M^T)^T$, where each block $y_j \in \{0,1\}^h$ corresponds to $\test{Z^{(j)}}{x^*}$, and hence is given as 
\[
y_j = \bigvee_{i \in \supp{x^*}} Z^{(j)}_i = \bigvee_{i \in \supp{x^*} \cap (i,j)\in E} U_i,
\]
where, the latter equality follows from the construction of $Z^{(j)}$.

The decoding algorithms of \cite{SAFFRON} relies on the following theorem that shows that each defective item $i$ appears by itself with high probability on at least one right node. Such a right node is called a singleton node. 
\begin{definition}[Singleton Node]
    A right node that is connected to one and only one defective item is called a singleton.
\end{definition}

In particular, the authors show the following lemma which we extend for any set of $k$ defectives using a union bound. 
\begin{lemma}[Thm 4.2 of \cite{SAFFRON}]
    Let $G = (L, R, E)$ be a random bipartite graph with $M=e(k^2\log{n} + k(1+\alpha)\log{k})$ right nodes, where each left node is connected to a certain right node with probability $p = \frac{1}{k}$ independently. Let $S^* \subseteq L$ be any set of $k$ left vertices. Then each $i \in S^*$ is connected to at least one singleton node with probability at least $1-k^{-\alpha}$. 
\end{lemma}

\cite{SAFFRON} show that each singleton node can be identified efficiently for an appropriate choice of signature matrix $U$, and the defective item can be decoded from the output of the tests $y_i$'s. \cite{guruswami2023noise} later showed that the number of rows in the testing matrix can be optimized further using Reed-Solomon codes to construct the signature matrix which also makes the testing scheme robust to random bit-flip noise. 
In this work, we show that a careful selection of the signature matrix can also make the above testing procedure robust to a small set of adversarial deletions. 

Let $C \subseteq \mathbb{F}_2^N$ be a linear insertion-deletion code of rate $K$ that can be efficiently decoded from a constant fraction of deletions. Existence and efficient construction of such binary linear codes was established by \cite{linear_const_insdel_code:LIPIcs.ICALP.2023.41}. 

\begin{theorem}[Thm 10 of \cite{linear_const_insdel_code:LIPIcs.ICALP.2023.41}]\label{thm:good_del_code}
For every constant $\gamma \in (0,1/3)$, there exists a binary linear insertion-deletion code with rate $\geq \frac13 -\gamma$, relative decoding radius $\Omega\left(  \frac{\gamma}{\log{\frac{1}{\gamma}}}\right)$, where both encoding and decoding are in polynomial time.
\end{theorem}

\textbf{Setting Parameters: }
We let $C$ to be the instantiataion of such codes with $\gamma = 1/8$ (arbitrary constant $<1/3)$, $K = \log n$, and $N = \Theta(\log n)$. Consider the codebook obtained by enumerating all the codewords of $C$ as columns of the matrix $U' \in \{0,1\}^{N \times n}$.
The signature matrix $U$ is then constructed by stacking $U'$ with its complement $\overline{U'}$. 
$U$ therefore, has $h= 2N = \Theta(\log n)$ rows. The testing matrix obtained by the construction described above using a bipartite graph with $M$ left vertices will have $m = Mh = O( k^2\log^2 n)$ rows. We now provide a decoding algorithm that recovers the set of $k$ defective items indexed by $S^* \subseteq [n]$ even after $\Delta = O(\log n)$ deletions.

\begin{algorithm}[h!]
\begin{flushleft}
\textbf{Input}: $\mathbf{A} \in \{0,1\}^{m \times n}$ --  Constructed as described above \\
$\ti{\mathbf{y}} = \corr(\test{A}{x^*}, \Delta) \in \{0,1\}^{(m-\Delta)}$\\
\end{flushleft}
\begin{algorithmic}[1]
\caption{Decode with Singletons}\label{alg:singleton-dec}
\State {Let $z \in \{0,1\}^m$ be the extension of $\tilde{y}$ obtained by padding $\Delta$ $0$'s at the end}
\State{Let $\hat{S} = \emptyset$}
\State{For each $i \in [M]$, let $z_i = z[(i-1)h +1: ih]$}
\ForEach{$i \in [M]$} 
    \If{$\wt(z_j) \le h/2 + \Delta$}
        \State {$\hat{S} \leftarrow \text{Decode}(z_j[:h/2], C$)}
    \EndIf 
\EndFor
\State {Return $\hat{S}$}
\end{algorithmic}
\end{algorithm} 

The correctness of Algorithm~\ref{alg:singleton-dec} follows from the following three lemmas: 

\begin{lemma}\label{lem:singleton-lem1}
    If the node $j \in R$ is a singleton node, then $\wt(z_j) \le h/2 + \Delta$. 
\end{lemma}

\begin{lemma}\label{lem:singleton-lem2}
    If the node $j \in R$ is not a singleton node, then $\wt(z_j) \ge h/2 + 2\Delta$. 
\end{lemma}

\begin{lemma}\label{lem:singleton-lem3}
    If $j \in R$ is a singleton node, then $d_{del}(z_j[:h/2], C) \le \Delta$. 
\end{lemma}

\begin{proof}[Proof of Theorem~\ref{thm:singleton}]
The proof of Theorem~\ref{thm:singleton} follows from the correctness of Algorithm~\ref{alg:singleton-dec}.

First, note that Algorithm~\ref{alg:singleton-dec} runs in time $O(MT)$, where $M = O(\log n)$, and $T$ is time to decode $C$ which we know runs in  $\textsc{poly(N)} = \textsc{poly}(\log n)$ time. Hence Algorithm~\ref{alg:singleton-dec} is super-efficient.

From Lemma~\ref{lem:singleton-lem1}, Lemma~\ref{lem:singleton-lem2}, it follows that the \textsc{If} condition in Line 5 is satisfied only by singleton nodes. Furthermore, by Lemma~\ref{lem:singleton-lem3}, it follows that the number of errors in the corresponding word is bounded above by $2\Delta < D$. Therefore, the codeword (or signature) corresponding to the defective item connected to the singleton node can be decoded and identified correctly. 
\end{proof}

To complete the proof of Theorem~\ref{thm:singleton}, we now present the proof of the three lemmas stated above. 

\begin{proof}[Proof of Lemma~\ref{lem:singleton-lem1}]
    Firstly, note that for any $i \in [M]$,
    \begin{align}\label{eq:lem1}
        \wt(z_i) & 
        \le \wt(\tilde{y}_i) + \Delta, 
    \end{align}
    where the inequality follows from the fact that the total number of deletions are at most $\Delta$, and therefore, the increase in weight can be at most $\Delta$.

    Finally, note that since $i$ corresponds to a singleton node, the true output 
    $y_i = [c^T, \overline{c}^T]^T$ for some codeword $c \in C$. 
    Therefore, the true output has weight exactly $N$, and after deletions, $\wt(\tilde{y}_i) \le N$. The proof follows by combining this observation with Equation~\ref{eq:lem1}. 
\end{proof}

\begin{proof}[Proof of Lemma~\ref{lem:singleton-lem2}]
    The proof follows from the fact that for any non-singleton node, the true output will be the \OR-sum of $\ell \ge 2$ distinct codewords (and their complements). In particular, 
    \[
    y_i = \bigvee_{j \le [\ell]} \begin{bmatrix}
        c_j\\ \overline{c_j}
    \end{bmatrix}
    \]
    From the distance property of the code $C$, we can conclude that the true output will have $\wt(y_i)\ge N + 2d$ since,
    \begin{align*}
        \wt(y_i) &= \wt(\bigvee_{j \le [\ell]} c_j) + \wt( \bigvee_{j \le [\ell]} \overline{c_j}) \\
        &\ge \wt(c_1 \vee c_2) + \wt(\overline{c_1} \vee \overline{c_2}) 
        = 2 \sum_{i \in [n]} \mathbbm{1}_{c_1[i] \neq c_2[i]} + \sum_{i \in [n]} \mathbbm{1}_{c_1[i] = c_2[i]} \\
        &\ge 2d + (N-d) = N+d
    \end{align*}
    The penultimate inequality follows from the observation that for all the indices $i \in [N]$ such that $c_1(i) \neq c_2(i)$, will contribute $1$ to both $\wt(c_1(i) \vee c_2(i))$ and $\wt(\overline{c_1}(i) \vee \overline{c_2}(i))$. Furthermore, the indices where $c_1(i) = c_2(i)$ will contribute to either $\wt(c_1(i) \vee c_2(i))$ or to $\wt(\overline{c_1}(i) \vee \overline{c_2}(i))$. 

    Therefore, after deletions, the $\wt(z_i)$ can do down by at most $\Delta$. Hence for $d > 3\Delta$,  we have that $\wt(z_i) \ge N + d - \Delta > N + \Delta$.
    
    Finally, note that since $i$ corresponds to a singleton node, the true output 
    $y_i = [c^T, \overline{c}^T]^T$ for some codeword $c \in C$. 
    Therefore, the true output has weight exactly $N$, and after deletions, $\wt(\tilde{y}_i) \le N$. The proof follows by combining this observation with Equation~\ref{eq:lem1}. 
\end{proof}

\begin{proof}[Proof of Lemma~\ref{lem:singleton-lem3}]

Let $i \in R$ be a singleton node, and let $y_i = [c^T, \overline{c}^T]^T$ be the uncorrupted output for some codeword $c \in C$. Let $\tilde{c} := z_i[:h/2 - \Delta]$ denote the part of the corrupted block that is a definite substring of $c$ (since the total number of deletions are at most $\Delta$). Therefore,  $LCS(c, \tilde{c}) \ge N - \Delta$, and $d_{del}(c, \tilde{c}) < \Delta \le d/3$. 

\end{proof}

\section{Deterministic construction of $(k,\Delta)$-deletion disjunct matrix}\label{sec:disjunct-deterministic}

In \cite{kautz1964nonrandom}, Kautz and Singleton (KS) provided a deterministic construction of $k$-disjunct matrices with $O(k^2 \log^2 n)$ rows using MDS codes, such as the Reed-Solomon codes, of appropriate parameters. We now show that using certain strong variants of MDS codes, a simple modification of the KS construction will give us $(k, \Delta)$-deletion disjunct matrices. 

Before we describe our construction, we will briefly review the KS construction of $k$-disjunct matrices. Fix any arbitrary ordering of the elements of $\mathbb{F}_q$. 
Given a $[N,K,D]_q$ - code $C \subseteq \mathbb{F}_q^N$, the KS construction first maps each symbol $i \in \mathbb{F}_q$ of the codeword to a binary vector $e_i \in \{0,1\}^q$, where $e_i$ is the standard unit vector with $1$ in the $i$-th position and $0$ everywhere else. This \emph{identity} map allows us to map each codeword $c \in C$ to a binary vector of length $qN$. \cite{kautz1964nonrandom} then showed that the $qN \times |C|$ binary matrix obtained by taking the binary representation of each codeword in $C$ is $k$-disjunct for any $k < N/(N-D)$. 

In this section, we show that if the code $C$ also possesses $(\Delta, d)-\ell_\infty$ distance property, then a slight modification of the above-described KS construction will give us a $(k, \Delta)$-deletion disjunct matrix. 



%

\begin{proof}[Proof of Theorem~\ref{thm:kddel}]
    For the construction of the $(k, \Delta)$-deletion disjunct matrix, we will employ the same identity map as in \cite{kautz1964nonrandom} but pad $\Delta$ zeros between any two symbols. 

    For each codeword symbol $\alpha_i \in \mathbb{F}_q$, construct a binary vector $v(\alpha_i) \in \{0,1\}^{q+\Delta}$ as 
    \begin{equation*}
            v(\alpha_i)[j] = \begin{cases}
                1 &  j = i + \Delta \\
                0 & \text{otherwise}.
            \end{cases}
        \end{equation*}

    Now, each codeword $c \in C$ can be mapped to a binary vector in $\{0,1\}^{N(q+\Delta)}$ by mapping each symbol to $\{0,1\}^{q+\Delta}$ using $v$ as $v(c) = (v(c_i))_{i \in [N]}$. Each vector $v(c)$ is, therefore, made of $N$ blocks each of length $q+\Delta$. Note that each block has exactly $1$ one, and the rest are zeros. 

    Let $M_C$ be the $N(q+\Delta) \times |C|$ binary matrix where each column corresponds to the binary vector obtained by mapping each codeword $c \in C$ using the above-defined \emph{padded-identity map}, i.e, 
    $M_C = [v(c_i)]_{i\in [N], c\in C}$. We now show that $M_C$ is $(k, \Delta)$-deletion disjunct if $C$ has $(\Delta, d)-\ell_\infty$ distance property.

    First, consider any two columns $x, y \in M_C$ corresponding to codewords $c_1, c_2 \in C$. Let $x', y' \in \{0,1\}^{Nq + (N-1)\Delta}$ be the binary vectors obtained after deleting some distinct set of $\Delta$ entries each from $x$ and $y$ respectively. 

    Partition the vectors $x', y'$ into $N$ blocks, where the first $N-1$ blocks are each of length $q+\Delta$, and the last block is of length $q$. Firstly, note that 
    the non-zero entries from a block in $x$ (or $y$) either gets deleted or is shifted upward by at most $\Delta$ places. Therefore, due to the padding of $\Delta$ zeros at the beginning of each block, no $1$ from a block in $x$ (or $y$) moves to another block. 
    Secondly, since $x$ has at least $D$ non-zero blocks, after $\Delta$ deletions, $x'$ will have at least $D-\Delta$ non-zero blocks. 
    
    Consider the number of blocks in $x'$ and $y'$ that have a $1-0$ match, i.e, $x'_i = 1$, and $y'_i = 0$. Since each non-zero block of $x'$ has only $1$ non-zero entry, we just have to count the number of non-zero blocks of $x'$ that are not identical to $y'$ after deletions. Using the $(\Delta, d)-\ell_\infty$ distance property of $C$, we know that there are at least $d$ entries where the $1$'s in $x$ and $y$ are separated by at least $\Delta$ positions. Therefore, even after deleting $\Delta$-zeros, the blocks corresponding to these $d$ entries (if non-zero) remain distinct. 
    Since at least $d-\Delta$ of these blocks in $x'$ are non-zero, it follows that there are at least $d-\Delta$ blocks with $1-0$ matches between $x'$ and $y'$. In other words, there are at most $N-(d-\Delta)$ blocks with no $1-0$ matches, 

    We now extend this observation to a larger set of $k$ columns. Let $x, y_1, \ldots, y_k$ be a set of $k+1$ columns of $M_C$. Using the same argument as above, we know that there are at most $N-(d-\Delta)$ blocks with no $1-0$ matches between each pair of $x'$ and $y'_i$, $i \in [k]$. If these blocks are all non-overlapping, we can conclude that there are at most $k(N-(d-\Delta))$ blocks with no $1-0$ matches between $x'$ and $\bigvee_i y_i'$. 

    Therefore, if $k(N-(d-\Delta)) < N$, then there is at least $1$ block with a $1-0$ match between $x'$ and $\bigvee_i y_i'$. 
\end{proof}

We further present a simple construction of a $(N,K,D)_{q'}$ code with $(\Delta, d)-\ell_\infty$ property from a $(N, K, D)_q$ code. 


\begin{proof}[Proof of Theorem~\ref{thm:code}]
    The construction follows by scaling of each coordinate of the codewords in $C$ by a factor of $\Delta$. This ensures that wherever the two codewords in $C$ differ by even $1$, they now differ by least $\Delta$. 

    Fix an arbitrary ordering of elements of $\F_q$ as $\alpha_1 < \alpha_2 < \ldots < \alpha_q$, and similarly consider any ordering of elements of $\F_{q'}$ as $\beta_1 < \beta_2 < \ldots < \beta_{q'}$. 

   Let $f:\mathbb{F}_q \rightarrow \mathbb{F}_{q'}$ be defined as follows: 
   \[
   f(\alpha_i) =  \beta_{\Delta \cdot i}
   \]

    The code $C'$ is then obtained from $C$ by applying the function $f$ coordinate-wise to codeword, \[
    C' = \{ (f(c[1]), f(c[2]), \ldots, f(c[N])) ~|~ c \in C \}
    \]

    Since $f$ preserves Hamming distance, $C'$ is also an $(N, K, D)$ code. 
    Also, for any two codewords, $c, c' \in C'$, we have that $|\psi(c[i]) - \psi(c'[i])| > \Delta$, for all $i \in \supp{c-c'}$. Recall that $\psi: \F_{q'} \rightarrow \mathbb{Z}$ defined as $\psi(\beta_i) = i$ according to the ordering fixed above. 
\end{proof}

Using this construction from Theorem~\ref{thm:code} with any MDS codes such as Reed Solomon codes with parameters $[N=\tilde{O}(k), K=\tilde{O}(\log n), D=N-K+1]$, along with Theorem~\ref{thm:kddel}, we get a $(k, \Delta)$-deletion disjunct matrix with $m = O(k^2 \Delta \log{n})$ rows that can tolerate roughly $O(\sqrt{\log n})$ deletions.

Note that the code construction in Theorem~\ref{thm:code} may not be optimal, and hence we get the same bounds on the number of tests as a trivial construction presented in Theorem~\ref{thm:del-dist}. We leave the intriguing problem of constructing better codes with $(\Delta, d)-\ell_\infty$ property as an open problem.

\section{Conclusion and Open Problems}
In this work, we initiate the study of non-adaptive combinatorial group testing in the presence of deletions. We provide both necessary and sufficient conditions for the exact recovery of the underlying set of defectives. Furthermore, we also provide a randomized and a potential deterministic construction of testing matrices that satisfy the necessary conditions and give an efficient decoding algorithm. 


Furthermore, obtaining explicit constructions of non-trivial deletion disjunct matrices with fewer tests is left for future work. One suggested approach is to use a Kautz-Singleton-like construction using codes with $(\Delta, d)-\ell_\infty$ distance properties, however, constructing such codes with good parameters is left as an open problem which can be of independent interest.  



\clearpage
\printbibliography

\clearpage
\appendix

\end{document}